\documentclass[11pt]{article}
\usepackage{amsthm,amsmath}

0.1cm \evensidemargin - 5cm

\newtheorem{thm}{Theorem}[section]
\newtheorem{prop}[thm]{Proposition}
\newtheorem{lem}[thm]{Lemma}
\newtheorem{cor}[thm]{Corollary}

\newtheorem{remark}[thm]{Remark}
\begin{document}
\title{\Large{{\bf The limit distribution in the $q$-CLT for $q \ge 1$ is unique and can not have a compact support}}}
\author{Sabir Umarov$^{1}$ 
and Constantino Tsallis$^{2,3}$}
\date{}
\maketitle
\begin{center}
$^{1}$ {\it   Department of Mathematics and Physics, University of New Haven, 300
Boston Post Road, West Haven, CT 06516, USA  \\} 
$^{2}$ {\it Centro Brasileiro
de Pesquisas Fisicas, and National Institute of Science and Technology for Complex Systems,
Xavier Sigaud 150, 22290-180 Rio de Janeiro-RJ, Brazil}\\
$^{3}$ {\it Santa Fe Institute,
1399 Hyde Park Road, Santa Fe, NM 87501, USA}
\end{center}
\begin{abstract}
{ In a paper by Umarov, Tsallis and Steinberg (2008),
a generalization of the Fourier transform, called the $q$-Fourier transform, was introduced and applied for the proof of a $q$-generalized central limit theorem ($q$-CLT). 
Subsequently, 
Hilhorst illustrated  (2009 and 2010) that the
$q$-Fourier transform for $q>1$, 
is not invertible in the space of density functions. Indeed, using an invariance principle, he constructed a
family of densities with the same $q$-Fourier transform and noted
that "as a consequence, the $q$-central limit theorem
falls short of achieving its
stated goal". The distributions constructed there have compact support. We prove now that the limit distribution in the $q$-CLT is unique and can not have a compact
support.  This result excludes all the possible counterexamples which can be constructed using the invariance principle and fills
the gap mentioned by Hilhorst.}
\end{abstract}
\renewcommand{\baselinestretch}{2}
\section{Introduction}
The $q$-central limit theorem ($q$-CLT), proved in
\cite{UmarovTsallisSteinberg2008} (see also
\cite{UmarovTsallisGellMannSteinberg2010}),  deals with sequences of random variables of the form
\begin{equation} 
\label{ZN} 
Z_N = \frac{S_N-N\mu_q}{N^{\frac{1}{2(2-q)}}},
\end{equation}
where $S_N=X_1+ \dots +X_N$, {the} random variables $X_1, \dots X_N$
{being} identically distributed and  $q$-independent, 
$\mu_q=\int x [f(x)]^qdx,$ and $1\le q <2.$ Here $f(x)$ is the density function of the random variable $X_1.$
Without loss of generality one can assume that $\mu_q=0$. Three types of
$q$-independence were discussed in paper
\cite{UmarovTsallisSteinberg2008}. Namely, identically distributed
random variables $X_N$ are $q$-independent (see
\cite{UmarovTsallisSteinberg2008}) if
\begin{align} \label{q-ind1}
\mbox{Type I:} ~~~~
&F_q[X_1+\dots+X_N](\xi)=F_q[X_1](\xi)\otimes_q\dots\otimes_q
F_q[X_N](\xi);\\ \label{q-ind2} \mbox{Type II:} ~~~~
&F_{q}[X_1+\dots+X_N](\xi)=F_q[X_1](\xi)\otimes_{q_1}\dots\otimes_{q_1}
F_q[X_N](\xi); \\ \label{q-ind3} \mbox{Type III:} ~~~~
&F_q[X_1+\dots+X_N](\xi)=F_{q_1}[X_1](\xi)\otimes_{q_1}\dots\otimes_{q_1}
F_{q_1}[X_N](\xi),
\end{align}
if these relationships hold for all $N\ge 2$ and $\xi \in
(-\infty,\infty)$; $q_1=\frac{1+q}{3-q}$.
Here the operator $F_q$ is the $q$-Fourier transform ($q$-FT) defined
as
\begin{equation}
\label{Fq} F_q[X_1](\xi) = \tilde{f}_{q}(\xi):= \int_{-\infty}^\infty
\frac{f(x)\,dx}{[1+i(1-q)x\xi f^{q-1}(x)]^{\frac{1}{q-1}}} \,,
\end{equation}
{with $q>1$}. If $q \to 1+0,$ then $F_q[X_1](\xi) \to F[X_1](\xi) =\int_{-\infty}^{\infty} f(x) e^{i x \xi} dx,$ coinciding with the Fourier transform of $f.$

The q-CLT states that if  $X_1, \dots X_N$
are identically distributed and  $q$-independent random variables, then the sequence $Z_N$ in \eqref{ZN} weakly converges to a random variable with the $q_{-1}$-Gaussian density; see \cite{UmarovTsallisSteinberg2008} for details.

The invertibility of $q$-FT in the class of $q$-Gaussian densities is established in \cite{UmarovTsallis2008} and in the space of hyper-functions in \cite{PlastinoRocca2012a,PlastinoRocca2012b}. However, using a specific invariance principle Hilhorst \cite{Hilhorst2010,Hilhorst2009} showed  that  $q$-FT is not invertible in the entire space of densities.  He constructed a family of densities containing the $q$-Gaussian and with the same $q$-FT. Any density of this family except the $q$-Gaussian has a compact support. In the present note we establish that a limit distribution in $q$-CLT  can not have a compact support. This fact implies that all the distributions with compact support in Hilhorst's counterexamples can not be a limiting distribution in the $q$-CLT, except the $q$-Gaussian density. However, deformations used in the invariance principle with functions $H(\xi), \, H(0)=0,$ lead to distributions,  which have noncompact support and share the same asymptotic behaviour at infinity as the $q_{-1}$-Gaussian. We prove that the limit distribution $Z_{\infty}$ and {any of} its $H$-deformation has the same $(2q-1)$-variance if and only if {the deforming function}  $H(\xi)$ is identically zero. {Using this fact and intrinsic properties of $q$-independent random variables we prove the uniqueness of the limit distribution of the scaling limit of $q$-independent random variables. This fact rules out all the possible counterexamples indicated by Hilhorst in his paper \cite{Hilhorst2010}. Thus, the $q$-FT is used only for the existence of limiting distribution, while intrinsic properties of $q$-independent random variables supply the uniqueness of this limiting distribution. We note that the inverse $q$-FT is nowhere required in the present proof.}

Now let us recall some facts about the $q$-algebra, {$q$-exponential} and
{$q$-logarithmic} functions. By definition, the {\it $q$-sum} of two numbers is defined as
$x \oplus_q y = x+y+(1-q)xy$. The $q$-sum is commutative, associative, 
recovers the usual summing operation if $q=1$ (i.e. $x \oplus_1 y = x+y$), and preserves 
$0$ as the neutral element (i.e. $x \oplus_q 0 = x$).  The {\it $q$-product} for $x,y$
is defined by the binary relation $x \otimes_q y = [x^{1-q}+y^{1-q}-1]^{{1}\over{1-q}}.$
This operation also commutative, associative, recovers the usual 
product when $q=1$, and preserves $1$ as the unity. The $q$-product is defined only when $x^{1-q}+y^{1-q} \ge 1$. The $q$-exponential and $q$-logarithmic functions are respectively defined as (see for instance \cite{UmarovTsallisSteinberg2008})
\begin{align*}
\exp_q(x)&=[1+(1-q)x]_+^{{1}\over {1-q}} \;\;\;(\exp_1(x)=e^x) \,,\\
\end{align*}
{and}
\begin{align*}
\ln_q (x)&= \frac{x^{1-q}-1}{1-q}, \, x>0 \;\;\;(\ln_1(x)=\ln x).
\end{align*}
It is easy to see (see  \cite{UmarovTsallisSteinberg2008}) that for the $q$-exponential, the relations
$\exp_q(x \oplus_q y) = \exp_q(x) \exp_q(y)$ and
$\exp_q{(x+y)}=\exp_q(x) \otimes_q \exp_q(y)$ hold. In terms of
$q$-log these relations can be equivalently rewritten as follows:
$\ln_q (x \otimes_q y)=\ln_q x + \ln_q y$, and $\ln_q (x y)=\ln_q x
\oplus_q \ln_q y.$ It follows from the definition of $q$-logarithm that if $1<q_1<q_2,$
then
\begin{align} \label{ln01}
&\ln_{q_{_{2}}}(x) \ge \frac{q_1-1}{q_2-1}
\ln_{q_{_1}}(x) ~ ~ \mbox{for\, all} ~ x>1.
\end{align}
For $q > 1$ the $q$-exponential is defined for all $x<\frac{1}{q-1}$
and blows up at the point $x=\frac{1}{q-1}.$ The $q$-exponential can
also be extended to the complex plane and it is bounded on the
imaginary axis: $|\exp_q{(iy)}| \leq 1.$ Moreover, $|\exp_q{(iy)}|
\rightarrow 0$ if $|y| \rightarrow \infty.$ Using the
$q$-exponential function, the $q$-FT of $f$ can be represented in
the form
\begin{equation}
\label{identity2} \tilde{f}_q(\xi) = \int_{-\infty}^{\infty}f(x)\,
\exp_q{(ix \xi [f(x)]^{q-1})}\,dx.
\end{equation}
We refer the reader to the papers
\cite{UmarovTsallisSteinberg2008,UmarovTsallisGellMannSteinberg2010,
TsallisQueiros2007,QueirosTsallis2007,UmarovTsallis2007,UmarovQueiros2010,NelsonUmarov2010,JaureguiTsallis2010a,ChevreuilPlastinoVignat2010,JaureguiTsallis2011,JaureguiTsallisCurado2011} for various properties and applications of the $q$-FT. Also, functions of the form $\exp_q{(-\beta x^2)}$ ($\beta >0$) will be hereafter referred to as $q$-Gaussians.

At this point, before addressing the technical aspects of the present problem, let us remind why the $q$-CLT may be very relevant in physics and other disciplines. It is common belief that the ubiquity of Gaussians in nature and elsewhere is due to the classical CLT. Indeed, this theorem provides a mathematical basis for observing the Gaussian attractors under quite general circumstances involving many {\it independent} (or quasi-independent) random variables. Analogously, also $q$-Gaussians emerge ubiquitously in nature and elsewhere, which strongly suggests the existence of a wide class of many {\it correlated} random variables whose corresponding attractors are $q$-Gaussians instead of Gaussians. Such experimental and theoretical examples include anomalous diffusion in type-II superconductors \cite{AndradeSilvaMoreiraNobreCurado2010} and granular matter \cite{CombeRichefeuStasiakAtman2015}, non-Gaussian momenta distributions for cold atoms in optical lattices \cite{Lutz2003,DouglasBergaminiRenzoni2006,LutzRenzoni2013}, dirty plasma \cite{LiuGoree2008}, trapped atoms \cite{DeVoe2009},  area-preserving maps \cite{TirnakliBorges2016}, high-energy physics \cite{WilkWong2013}, probabilistic models \cite{BergeronCuradoGazeauRodrigues2016}, to mention but a few (see \cite{Biblio}).

\section{On the support of the limit distribution}
For the sake of simplicity we consider a continuous and symmetric
about zero density function $f$ of a random variable $X_1.$ Other
cases can be considered in a similar manner with appropriate care.
Denote $\lambda(x)=x[f(x)]^{q-1},$ where $1 \le q<2.$ Since $f$ is
symmetric, it suffices to consider $\lambda(x)$ only for positive
$x.$ Suppose the maximum value of $\lambda$ is $m$ and
$x_m>0$ is the rightmost point where $\lambda$ attains its maximum, i.e. $m=\max_{0<x \le a}\{x[f(x)]^{q-1}\}=x_m [f(x_m)]^{q-1}. $ Let
\[
\tau_{\ast} = 
\begin{cases} 
\frac{1}{m(q-1)}, &  \mbox{if } 0 <q<2, \\
\infty, & \mbox{if } q= 1. 
\end{cases} 
\]
\begin{prop}\label{Prop:PW}
Let $f$ be a continuous symmetric density function with $supp \, f
\subseteq [-a,a].$ Then the $q$-FT of $f$ satisfies the following
estimate
\begin{equation} \label{prop1:est}
|\tilde{f}_q(\eta-i \tau)| \le \exp_q(x_m M_q \tau),
\end{equation}
where $\eta \in (-\infty,\infty),$ $\tau < \tau_{\ast},$ $M_q=
\max_{[0,a]}\{[f(x)]^{q-1}\},$ and $x_m$ is the rightmost point where
$xf^{q-1}$ attains its maximum $m.$
\end{prop}
\begin{proof} For $f$ with $supp \, f \subseteq [-a,a],$ equation
(\ref{Fq}) takes the form
\begin{equation}
\label{Fq10} \tilde{f}_{q}(\xi)= \int_{-a}^a
\frac{f(x)dx}{\left[1+i(1-q)x\xi f^{q-1}(x)\right]^{\frac{1}{q-1}}}.
\end{equation}
Let $\xi=\eta+i\tau$ where $\eta=\Re(\xi)$ is the real part of $\xi$
and $\tau=\Im(\xi)$ is its imaginary part. We assume that $\eta \in
(-\infty,\infty)$ and $|\tau| < \frac{1}{m(q-1)}.$ Then for the
denominator of the integrand in (\ref{Fq10}) one has
\begin{align} \label{denom}
\Big[ 1+i(1-q)x(\eta &-i\tau) f^{q-1}(x) \Big]^{\frac{1}{q-1}}= \Big[1+i(1-q)\eta
f^{q-1}(x)+(1-q)\tau x f^{q-1}(x) \Big]^{\frac{1}{q-1}} \notag \\&
= \Big[1+(1-q)\tau x f^{q-1}(x)\Big]^{\frac{1}{q-1}} \Big[1+i\frac{(1-q)\eta
f^{q-1}(x)}{1-(1-q)\tau x f^{q-1}(x)} \Big]^{\frac{1}{q-1}} \notag\\
&= \left(\exp_q(\tau x f^{q-1}(x))\right)^{-1}\left(\exp_q(
i\frac{(1-q)\eta f^{q-1}(x)}{1-(1-q)\tau x f^{q-1}(x)}
)\right)^{-1}.
\end{align}
Using the inequality $|\exp(iy)| \le 1$ valid for all $y \in
(-\infty,\infty)$ if $q>1,$ it follows from \eqref{denom} that
\begin{align*} 
\Big|1+i(1-q)x(\eta &+i\tau) f^{q-1}(x)\Big|^{\frac{1}{q-1}} \ge
\left(\exp_q(\tau x f^{q-1}(x))\right)^{-1},
\end{align*}
or
\begin{align} \label{denom10}
\Big|1+i(1-q)x(\eta &+i\tau) f^{q-1}(x)\Big|^{\frac{1}{1-q}} \le
\left(\exp_q(\tau x f^{q-1}(x))\right).
\end{align}
Now, \eqref{Fq10} together with \eqref{denom10} and $f(x)$ being a
density function, yield \eqref{prop1:est}. 
\end{proof}

\begin{remark}
Proposition \eqref{Prop:PW} can be viewed as a generalization of the
well known Paley-Wiener theorem. Indeed, if $q=1$ then
\eqref{prop1:est} takes the form
\begin{equation} \label{PW}
|\tilde{f}(\eta-i\tau)| \le \exp(a \tau), \, \eta+i\tau \in
\mathcal{C},
\end{equation}
which represents the Paley-Wiener theorem for continuous density
functions.
\end{remark}
Inequality \eqref{PW} can be used for estimation of the size of the
support of $f.$ Consider an example with
$f(x)=(2a)^{-1}{\mathcal{I}}_{[-a,a]}(x),$ where
$\mathcal{I}_{[-a,a]}(x)$ is the indicator function of the interval
$[-a,a].$ The Fourier transform of this function is
$\tilde{f}(\xi)=(a\xi)^{-1}\sin (a\xi),$ $M_q=M_1=1,$ and $x_m=a.$
Therefore, we have $|\tilde{f}(-i\tau)| \le e^{a\tau}, \tau >0.$
The latter yields 
\[
\displaystyle{2a \ge 2 \sup_{\tau >0} \frac{\ln|\tilde{f}(-i\tau)|}{\tau}},
\] 
which gives an estimate from below for the size $d(f)=2a$ of the support of $f.$
This idea can be used to estimate the size 
of the support of $f$ using the $q$-FT and Proposition
\ref{Prop:PW}. Namely, inequality \eqref{prop1:est} with $\eta=0$
implies
\begin{equation} \label{suppsize}
\displaystyle{d(f)=2a \ge 2x_m \ge \frac{2}{M_q} \sup_{0<\tau<\tau_{\ast}} \frac{\ln_q |\tilde{f}_q(-i\tau)|}{\tau}}.
\end{equation}
We notice that the integrand in the integral
\[
\tilde{f}_q(-i\tau)=\int_{-a}^{a}\frac{f(x)dx}{[1-(q-1)\tau x
f^{q-1}(x)]^{\frac{1}{q-1}}}
\]
is strictly grater than $f(x)$ if $\tau>0,$ implying
$|\tilde{f}_q(-i\tau)|>1,$ since $f$ is a density function.
Therefore, the right hand side of \eqref{suppsize} is positive and
gives indeed an estimate of the size of the support of $f$ from
below.

Let $f_N(x)=f_{S_N}(x)$ be the density function of
$S_N=X_1+\dots+X_N,$ where $X_1,\dots,X_N$ are $q$-independent
random variables with the same density function $f=f_{X_1}$ whose
support is $[-a,a].$ We show that the $q$-independence condition can
not reduce the support of $f_N$ to an interval independent of $N.$
More precisely, $d(f_N)$ increases at the rate of $N$ when $N \to
\infty.$

\begin{thm}\label{Prop:PW1}
Let  $X_1,\dots,X_N$ {be} $q$-independent of any type I-III random
variables all having the same density function $f$ with $supp \, f
\subseteq [-a,a].$ Then, for the size of the density $f_N$ of $S_N$,
there exists a constant $K_q >0$ such that the estimate
\begin{equation} \label{prop1:estfN}
\displaystyle{d(f_N) \ge K_q N  \sup_{0<\tau<\tau_{\ast}} \frac{\ln_q |\tilde{f}_q(-i\tau)|}{\tau}}
\end{equation}
holds.
\end{thm}

\begin{proof} Using formula \eqref{suppsize} one has
\begin{equation} \label{suppsize10}
\displaystyle{d(f_N) \ge \frac{2}{M_{q,N}}\sup_{0<\tau<\tau_{\ast}} \frac{\ln_q |\widetilde{({f}_{N})}_{q}(-i\tau)|}{\tau}},
\end{equation}
where $M_{q,N}=\max_{x\in[-Na,Na]} f_N^{q-1}(x).$ It is clear from
probabilistic arguments that $M_{q,N} \le M_q$ for all $N\ge 2.$
Therefore, it follows from \eqref{suppsize10} that
\begin{equation} \label{suppsize15}
\displaystyle{d(f_N) \ge \frac{2}{M_{q}}\sup_{0<\tau<\tau_{\ast}} \frac{\ln_q |\widetilde{({f}_{N})}_{q}(-i\tau)|}{\tau}},
\end{equation}
Let $X_N$ be $q$-independent of type I (see \eqref{q-ind1}). Making
use of the inequality $|z-r|\ge |z|-r,$ which holds true for any
complex $z$ and positive integer number $r,$ one has
\begin{align*}
\left|\widetilde{({f}_{N})}_{q}(-i\tau)\right|&= \left|\tilde{f}_q(-i\tau) \otimes_q
\dots \otimes_q \tilde{f}_q (-i \tau)\right|
\\
&=
\left|[N
\left(\tilde{f}_q(-i\tau)\right)^{1-q}-(N-1)]^{\frac{1}{1-q}} \right|
\\&\ge \left[N
\left| \tilde{f}_q(-i\tau) \right|^{1-q}-(N-1)\right]^{\frac{1}{1-q}}\\&= \left|\tilde{f}_q(-i\tau) \right|
\otimes_q \dots \otimes_q \left|\tilde{f}_q (-i \tau) \right|.
\end{align*}
Taking $q$-logarithm of both sides in this inequality and using the
property $\ln_q(g \otimes_q h)=\ln_q g + \ln_q h,$ one obtains
\begin{equation}
\label{suppsize16} \ln_q \left|\widetilde{({f}_{N})}_{q}(-i\tau)\right| \ge N
\ln_q \left|\tilde{f}_q(-i\tau)\right|.
\end{equation}
Now estimate \eqref{prop1:estfN} follows from inequalities
\eqref{suppsize15} and \eqref{suppsize16}.

Similarly, for random variables independent of type III, we have
\begin{equation}
\label{suppsize30} \ln_{q} \left|\widetilde{({f}_{N})}_{q}(-i\tau)\right| \ge N
\ln_{q_{_1}} \left|\tilde{f}_{q_{_1}}(-i\tau) \right|.
\end{equation}

For random variables $X_N$ independent of type II, equation
\eqref{suppsize16} takes the form
\begin{equation}
\label{suppsize20} \ln_{q_{_1}} \left|\widetilde{({f}_{N})}_{q}(-i\tau)\right|
\ge N \ln_{q_{_1}} \left|\tilde{f}_q(-i\tau)\right|.
\end{equation}
Since $1<q<q_1$ and $\frac{q-1}{q_1-1}=\frac{3-q}{2},$ making use of
inequalitiy \eqref{ln01},
and taking into account that
$|\widetilde{({f}_{N})}_{q}(i\tau)| \ge 1,$ one has
\begin{equation*} 
\ln_{q_{_1}} \left| \widetilde{({f}_{N})}_{q}(-i\tau) \right| \ge
\frac{(3-q)}{2} \ln_{q} \left| \widetilde{({f}_{N})}_{q}(-i\tau) \right|,
\end{equation*}
which implies
\begin{equation}
\label{suppsize21} \ln_q \left| \widetilde{({f}_{N})}_{q}(-i\tau) \right| \ge \frac{2 N}{3-q}
\ln_q |\tilde{f}_q(-i\tau)|.
\end{equation}

Both \eqref{suppsize30} and \eqref{suppsize21} obviously imply
estimate \eqref{prop1:estfN}. 
\end{proof}

\begin{cor}
\label{cor1} Let  $X_1,\dots,X_N$ be $q$-independent of any type
I-III random variables all having the same density function $f$ with
$supp \, f \subseteq [-a,a].$ If the sequence $Z_N$ has a
distributional limit random variable in some sense, then this random
variable can not have a density with compact support. Moreover, due
to the scaling present in $Z_N,$ the support of the limit variable
is the entire set of real numbers.
\end{cor}
The proof obviously follows immediately from \eqref{prop1:estfN}
upon letting $N \to \infty.$

\section{On the variance and quasivariance of a limit distribution}
Let $1 \le q <2$ and a random variable $X$ with a density function $f(x)$ has zero $q$-mean ($\mu_q(X)=0$) and a finite quasivariance 
\begin{equation} \label{qv}
QV(X)=\nu_{2q-1}(X) \sigma_{2q-1}^2(X) =\int (x-\mu_q)^2
[f(x)]^{2q-1}dx,
\end{equation} 
where  $\nu_{2q-1}(X), \ \mu_q,$ and $\sigma_{2q-1}(X)$ are defined as 
\begin{equation} \label{nu000}
\nu_q=\nu_q(X)=\int[f(x)]^q dx, \ \mu_q=\int x [f(x)]^q dx,
\end{equation}
and
\begin{equation} \label{sigma000}
\sigma_{2q-1}^2=\sigma_{2q-1}^2(X)=\int (x-\mu_q)^2
\frac{[f(x)]^{2q-1}}{\nu_{2q-1}}dx. 
\end{equation}
Note that if $q=1$ then $\nu_{q}=1$ and $\sigma_{2q-1}^2 =Var(X),$ the variance of $X,$ implying $QV(X)=Var(X).$ As above, without loss of generality, we assume that $\mu_q=0.$
If $X$ and $Y$ are $q$-independent (of any type I-III) random variables, then for their quasivariances the relation
\begin{equation} \label{qv_100}
QV(X+Y)=QV(X)+QV(Y)
\end{equation}
holds. 
To see the validity of this fact one can use the formula $(F_q[X])^{''}(0)=-q QV(X)$ (see \cite{TsallisPlastinoAlvares-Estrada2009}) and the definition of $q$-independence \eqref{q-ind1}-\eqref{q-ind2}.
Taking into account that the density of $aX$ for a constant $a>0$ is $a^{-1}f(x/a),$ one can easily verify that  \cite{UmarovTsallisSteinberg2008}
 \begin{equation} \label{sigma001}
\sigma_{2q-1}(aX)={a^2} {\sigma_{2q-1}(X)}.
\end{equation}

Let $X_1$ be a random variable with the $q$-Gaussian density 
$$
G_q(\beta, x)=\frac{\sqrt{\beta}}{C_q}e_q^{-\beta x^2}, \, \beta>0,
$$ 
where $C_q$ is the normalizing constant \cite{UmarovTsallisSteinberg2008}. The direct calculation shows that $\sigma_{2q-1}^2(X_1)=\beta^{-1}.$ The sequence of identically distributed $q$-Gaussian variables $X_1,\dotso,X_N$ is $q$-independent of type II if the density of $X_1+\dots+X_N$ is the $G(N^{-\frac{1}{2-q}}\beta, x);$ (see \cite{UmarovTsallisSteinberg2008}). Using \eqref{sigma001}, one can see    
that
\begin{equation} \label{sigma}
\sigma_{2q-1}^2(X_1+\dots+X_N) = N^{\frac{1}{2-q}}\sigma_{2q-1}^2(X_1),
\end{equation}
and
\begin{align} 
\sigma_{2q-1}^2(Z_N)&=\frac{1}{N^{\frac{1}{2-q}}} \sigma_{2q-1}^2(X_1+\dots+X_N) 
={\sigma_{2q-1}^2(X_1)},   
   \quad \mbox{for all} \ N \ge 1.   \label{sigma_zn}
\end{align}

If $q=1$ then \eqref{sigma} reduces to the known relationship $Var (X_1+\dots+X_N)=NVar(X_1)$ valid for variances of independent and identically distributed (i.i.d.) random variables $X_1,\dotso,X_N.$ In this case  \eqref{sigma_zn} becomes $Var(Z_N)=Var(X_1),$ valid for rescaled sums of i.i.d. random variables. 
In other words the equality 
\begin{equation}
\label{qv000}
QV(Z_N)=QV(X_1)
\end{equation}
holds if $q=1.$

We notice that relations \eqref{qv_100} and \eqref{sigma}  imply
\begin{equation} \label{nu}
\nu_{2q-1}(X_1+\dots+X_N) = \frac{\nu_{2q-1}(X_1)}{N^{\frac{q-1}{2-q}}}.
\end{equation}
Indeed, due to \eqref{qv_100} 
\begin{align*}
\sigma_{2q-1}^2(X_1+\dots+X_N) &=  \frac{QV(X_1+\dots + X_N)}{\nu_{2q-1}(X_1+\dots+X_N)} 
\\
&=\frac{QV(X_1)+\dots + QV(X_N)}{\nu_{2q-1}(X_1+\dots+X_N)} 
\\
&= \frac{N QV(X_1)}{\nu_{2q-1}(X_1+\dots+X_N)} 
\\
&=N \frac{\nu_{2q-1}(X_1) \sigma_{2q-1}^2(X_1)}{\nu_{2q-1}(X_1+\dots+X_N)}.
\end{align*}
The latter and equality \eqref{sigma} imply \eqref{nu}.

In the case $q=1$ for any i.i.d. random variables $\nu_{2q-1}(X_1+\dots+X_N)=\int f_{_N}(x)dx=1,$ where $f_{_N}$ is the density function of the sum $X_1+\dots+X_N.$ However, if $q>1$ then $\nu_{2q-1}(X_1+\dots+X_N)$ does depend on $N,$ and as relation \eqref{nu} shows, the most natural dependence on $N$ can be given by
 the condition
\begin{equation} \label{nu_hyp}
\nu_{2q-1}(X_1+\dots+X_N) \sim O \left( N^{\frac{1-q}{2-q}}\right), \ N \to \infty.   
\end{equation}
Let us consider some examples. If $q=1$ then for any i.i.d. sequence of random variables the relation \eqref{nu} is reduced to the identity $1=1,$ thus satisfying condition \eqref{nu_hyp}. As the above example states, for the type II $q$-i.i.d. $q$-Gaussian random variables relation \eqref{nu} is valid, thus again satisfying condition \eqref{nu_hyp}. One can verify that for $q$-Gaussians independent of type I or III  condition \eqref{nu_hyp} is also verified. Random variables studied in \cite{VignatPlastino2007} also satisfy condition \eqref{nu_hyp} since they are asymptotically equivalent to $q$-indpendent random variables (see \cite{VignatPlastino2007}). As is shown in \cite{HahnJiangUmarov2010} random variables in \cite{VignatPlastino2007} are variance mixtures of normal densities. This gives a strong evidence of the fact that the subclass of variance mixtures of normal densities leading to $q$-Gaussians will also satisfy \eqref{nu_hyp}. For connection of variance mixtures to superstatistics developed by Beck and Cohen \cite{BeckCohen2003} see \cite{HahnJiangUmarov2010}. In our further considerations we assume condition \eqref{nu_hyp} for $q$-i.i.d. random variables $X_1, \dots, X_N.$

The asymptotic expansion of the $q$-exponential function $\exp_q(x)$ near zero
implies that (see Proposition II.3 in \cite{UmarovTsallisGellMannSteinberg2010}, case $\alpha=2$)
\begin{equation} \label{qv}
F_q[X](\xi)=1-\frac{q}{2}QV(X)\xi^2+o(\xi^2), \, \xi \to 0.
\end{equation}
Making use of properties of the $q$-Fourier transform one can see
that (see details in \cite{UmarovTsallisSteinberg2008})
\begin{equation} \label{qv10}
F_q[Z_N](\xi)=1-\frac{q}{2}QV(X_1)\xi^2+o(\frac{\xi^2}{N}), \, N \to
\infty,
\end{equation}
which shows that $QV(Z_N)=QV(X_1), \, N \ge 1.$ Hence, relation \eqref{qv000} is valid not only for $q=1,$ but for all $1<q<2,$ as well. This immediately implies
that if the limit distribution $Z_{\infty}=\lim_{N\to \infty}Z_N$
exists in some sense, then its quasivariance must be equal to $QV(X_1)$,
i.e.
\begin{equation} \label{qv15}
QV(Z_{\infty})=QV(X_1).
\end{equation}

The lemma below will be used in Section \ref{section:uniqueness}.
\begin{lem}\label{lemma:omega}
Let $\omega(x)$ be a continuous function defined on $[0,\infty)$
such that
\begin{enumerate}
\item[(a)] $\omega(1)=0,$
\item[(b)] $\omega(x)>0$ on $(0,1),$ and $\omega(x)<0$ on $(1,\infty),$
\item[(c)] $\int_0^1 \omega (x)dx=\int_1^{\infty}|\omega(x)|dx.$
\end{enumerate}
Then
\begin{equation}\label{est10}
\int_0^1 x^2 \omega(x)dx<\int_1^{\infty}x^2 |\omega(x)|dx.
\end{equation}
\end{lem}
\begin{proof} Since $x^2<1$ for $x \in (0,1),$ one has
\begin{equation} \label{est1}
\int_0^1 x^2 \omega(x)dx < \int_0^1 \omega(x)dx.
\end{equation}
Similarly, for $x>1,$
\begin{equation} \label{est2}
\int_1^{\infty} |\omega(x)|dx < \int_1^{\infty}x^2 |\omega(x)|dx.
\end{equation}
Now condition $(c)$ and estimates \eqref{est1} and \eqref{est2} imply
\eqref{est10}. 
\end{proof}

\section{On the invariance principle and Hilhorst's counterexamples}
\label{section:invariance}
In this Section first we recall the
invariance principle used by Hilhorst \cite{Hilhorst2010} to
construct counterexamples which show that $q$-FT is not invertible.
Then we apply the invariance principle to the $q$-Gaussian and study
properties of densities produced by the invariance principle in this
case. Let $f(x), \, x \in (-\infty,\infty),$ be a symmetric density
function, such that $\lambda(x)=x[f(x)]^{q-1}$ restricted to the
semiaxis $[0,\infty)$ has a unique (local) maximum $m$ at a point
$x_m.$
In other words $\lambda(x)$ has two monotonic pieces,
$\lambda_{-}(x), \, 0 \le x \le x_m,$ and $\lambda_{+}(x), \, x_m
\le x < \infty.$ Let $x_{\pm}(\xi), \, 0 \le \xi \le m,$ denote the
inverses of $\lambda_{\pm}(x),$ respectively. Then the $q$-FT
($1<q<2$) of $f$ can be expressed in the form, see
\cite{Hilhorst2010}
\begin{equation*}
\label{qFT1}
\tilde{f}_q(\xi)=\int_{-\infty}^{\infty}F(\xi')\exp_q(i\xi\xi')d\xi',
\end{equation*}
{where}
\begin{equation}
F(\xi)=\frac{q-2}{q-1}\xi^{\frac{1}{q-1}}\frac{d}{d\xi} \left[ x_{-}^{\frac{q-1}{q-2}}(\xi)-
x_{+}^{\frac{q-1}{q-2}}(\xi) \right], \, \xi \in [0,m].
\end{equation}
Then the invariance principle yields
\begin{equation}
\label{qFT2}
F(\xi)=\frac{q-2}{q-1}\xi^{\frac{1}{q-1}}\frac{d}{d\xi}\left[{X}_{-}^{\frac{q-1}{q-2}}(\xi)-
X_{+}^{\frac{q-1}{q-2}}(\xi)\right], \, \xi \in [0,m],
\end{equation}
where
\begin{equation} \label{Xpm}
X_{\pm}^{\frac{q-1}{q-2}}(\xi)=x_{\pm}^{\frac{q-1}{q-2}}(\xi)+H(\xi),
\end{equation}
with $H(\xi)$ being a function defined on $[0,m],$ and such that
$X_{\pm}(\xi)$ 
are invertible. Denote by
$\Lambda (x)$ the function defined by the two pieces of inverses of
$X_{\pm}(\xi),$ namely
\begin{equation*}
\Lambda_{H} (x)=
\begin{cases}
X_{-}^{-1}(x),            &\text{if $0 \leq x \le x_{m,H}$,}\\
X_{+}^{-1}(x),            &\text{if $x > x_{m,H}$,}
\end{cases}
\end{equation*}
where
$x_{m,H}=[(q-1)^{\frac{q-1}{2(2-q)}}+H(m)]^{-\frac{2-q}{q-1}}.$ The
function $\Lambda_H(x)$ is continuous, since
$X_{-}^{-1}(x_{m,H})=X_{+}^{-1}(x_{m,H}).$ Then
\begin{equation}
\label{fH} f_H(x)=\left(\frac{\Lambda(x)}{x} \right)^{\frac{1}{q-1}}
\end{equation}
defines a density function with the same $q$-FT as of $f.$ The
density $f_H$ coincides with $f$ if $H(\xi)$ is identically zero.

Now assume that $f(x)$ is a $q$-Gaussian,
\begin{equation} \label{qgaussian}
f(x)=G_q(x)=\frac{C_q}{\left[1+(q-1)x^2\right]^{\frac{1}{q-1}}}, \, 1<q<2,
\end{equation}
where $C_q$ is the normalization constant.
Obviously, $G_q(x)$ is symmetric, and the function
$\lambda_q(x)=x[G_q(x)]^{q-1}$ considered on the semiaxis
$[0,\infty)$ has a unique maximum $m=\frac{C_q^{q-1}}{2\sqrt{q-1}}$
attained at the point
\begin{equation}\label{a}
x_m=(q-1)^{-\frac{1}{2}}.
\end{equation}
Moreover, the functions $x_{\pm}(\xi)$
in this case take the forms 
\begin{equation}
\label{pieces1} x_{\pm}(\xi)=\frac{C_q^{q-1}\pm
\left[C_q^{2(q-1)}-4(q-1)\xi^2 \right]^{1 \over 2}}{2\xi(q-1)}, \, 0 < \xi \le m.
\end{equation}
We denote the density $f_H(x)$  and the function $\Lambda_{H}(x)$
corresponding to the $q$-Gaussian by $G_{q,H}(x)$ and
$\Lambda_{q,H}(x),$ respectively.  Hilhorst, selecting $H(\xi)=A \ge
0$ constant, constructed a family of densities
\begin{equation}
\label{fA} G_{q, A}(x)=\frac{C_q
\left( x^{\frac{q-2}{q-1}}-A \right)^{\frac{1}{q-2}}}{x^{\frac{1}{q-1}} \left[1+(q-1)(x^{\frac{q-2}{q-1}}-A)^{2\frac{q-1}{q-2}} \right]^{\frac{1}{q-1}}},
\end{equation}
which
have the same $q$-FT as the $q$-Gaussian for all $A.$ The following
statement shows that none of the densities $G_{q, A}(x)$ can serve as the
limit distribution in the $q$-CLT, except the one, corresponding to
$A=0,$ which coincides with the $q$-Gaussian, $G_{q,0}(x)=G_q(x)$.

\begin{prop}
\label{prop:comsup} Let $H(0) > 0.$ Then the support of $G_{q,
H}(x)$ is compact, and 
$$supp \, G_{q,H}=
\left[-\Big(H(0)\Big)^{\frac{q-1}{q-2}},\Big(H(0)\Big)^{\frac{q-1}{q-2}}\right].
$$
\end{prop}

\begin{proof} Since $\lim_{\xi \to 0} x_{+}(\xi)=+\infty,$ the
largest value of $X_{+}$ is equal to $\lim_{\xi \to
0}X_{+}(\xi)=[H(0)]^{\frac{q-1}{q-2}}.$ Therefore, the inverse of
$X_{+}$ is defined on the interval $\left[x_0,
[H(0)]^{\frac{q-1}{q-2}} \right],$ where $x_0>0$ is some number obtained by
{\bf a} shifting of $x_m$ depending on $H(m).$ On the other hand the
smallest value of $x_{-}$ is zero, taken at $\xi=0.$ Therefore, the
inverse of $X_{-}$ is defined on the interval $[0,x_0].$ Hence,
 by symmetry, $G_{q,H}$ has the support
$\left[-[H(0)]^{\frac{q-1}{q-2}},[H(0)]^{\frac{q-1}{q-2}}\right].$
\end{proof}
\begin{remark}
{Note that $H(0)$ can not be negative. In fact, if $H(0)<0,$ then
{either} $X_{\pm}$ is not invertible or, if it is invertible, its
inverse does not define a density function.}
\end{remark}
{Proposition \ref{prop:comsup} implies that if $H(0) > 0$ then, due
to Corollary \ref{cor1}, $G_{q,H}(x)$ can not be the density
function of the limit distribution in the $q$-CLT. Thus none of the
densities in Hilhorst's counterexamples\footnote{See Examples 2 and
3 in \cite{Hilhorst2010}. Example 4 is not relevant to the $q$-CLT,
since in this case, $(2q-1)$-variance of the 2-Gaussian does not
exist, and consequently the $q$-CLT is not {applicable}.}, except the
$q$-Gaussian, can serve as an
attractor in the $q$-CLT.}

Only one possibility is left, {namely} $H(0)=0.$ The next
proposition establishes that, in this case, $G_{q,H}(x)$ is
asymptotically equivalent to $G_q(x) \equiv G_{q,0}(x).$

\begin{prop} \label{prop:asequiv}
Let $H(0)=0.$ Then
$$\lim_{|x|\to \infty} \frac{G_{q, H}(x)}{G_q(x)}=1. $$
\end{prop}

\begin{proof} {Since} $H(0)=0$, then obviously
\begin{equation*}
\lim_{\xi \to 0} \frac{X_{+}(\xi)}{x_{+}(\xi)}=\lim_{\xi \to 0}
\left(1+\frac{H(\xi)}{x_{+}(\xi)} \right)=1.
\end{equation*}
Therefore, for inverses one has
\begin{equation*}
\lim_{x \to +\infty} \frac{X_{+}^{-1}(x)}{x_{+}^{-1}(x)}=1.
\end{equation*}
This implies
\begin{equation*}
\lim_{x \to +\infty} \frac{G_{q,H}(x)}{G_q(x)}= \lim_{x \to +\infty}
\left(\frac{\frac{X_{+}^{-1}(x)}{x}}{\frac{x_{+}^{-1}(x)}{x}}\right)^{\frac{1}{q-1}}=1.  
\end{equation*}
\end{proof}

\begin{remark}
Propositions \ref{prop:comsup} and \ref{prop:asequiv} establish that
$G_{q,H}$ can identify a limiting distribution in the $q$-CLT only
if $H(0)=0.$ However, in this case, independently {from} other
values of $H(\xi),$ the density $G_{q,H}(x)$ is asymptotically
equivalent to the $q$-Gaussian, i.e. $G_{q,H}(x) \sim G_{q}(x) \, \,
\mbox{as} \, \, |x|\to\infty.$
\end{remark}

The statement of the following proposition can be proved exactly as
Proposition \ref{prop:asequiv}, replacing $X_+(\xi), \, x_+(\xi)$ by
functions $X_-(\xi), \, x_-(\xi),$ respectively.

\begin{prop} \label{prop:asequivat0}
Let $H(0)=0.$ Then
$$\lim_{x\to 0} \frac{G_{q, H}(x)}{G_q(x)}=1. $$
\end{prop}

\section{On the uniqueness of the limit distribution}
\label{section:uniqueness}
Let $X$ be a random variable with a symmetric density function $G$ and  let $G_{H}$ be the density function obtained from $G$ by $H$-deformation, where $H(\xi)$ is a continuous function such that $H(0)=0$ and does not change its sign on the interval $(0, x_m).$ Denote by $X_H$ the random variable corresponding to the density function  $G_{H}.$

\begin{lem} \label{lemma_equality}
Let $X$ and $X_H$ be random variables with the respective densities $G$  and $G_{H},$ 
and let $QV(X)=QV(X_H).$ Then 
$\sigma_{2q-1}^2(X)=\sigma_{2q-1}^2(X_H)$
if and only if $H(\xi)$ is identically zero.
\end{lem}
\begin{proof} {\it Sufficiency.} Let  $\sigma_{2q-1}^2(X)=\sigma_{2q-1}^2(X_H)$ and assume that $H(\xi)$ is not identically zero. This equality together with $QV(X)=QV(X_H)$ implies that $\nu_{2q-1}(X)=\nu_{2q-1}(X_H).$
Due to conditions on $H(\xi)$ both densities, $G$  and $G_{H},$  are symmetric, decreasing on the positive
semiaxis. Propositions \ref{prop:comsup} and \ref{prop:asequivat0}
imply that  $G(0)=G_H(0)б$ since $H(0)=0.$ Moreover, since both $G$ and
$G_{H}$ are densities there is a point $a>0$ such that $G(a)=G_{H}(a).$
Depending on the sign of $H(\xi),$ we have either
\begin{align}
\label{case_one}
G(x)>G_{H}(x) \quad &\mbox{on the interval} \quad (0,a) \quad \mbox{and} \quad G(x)<G_{H}(x) 
\\ 
\quad &\mbox{on the interval} \quad (a,\infty), \notag
\end{align}
\noindent or
\begin{equation}
\label{case_two} G(x)<G_{H}(x) \quad \mbox{on} \quad  (0,a)  \quad \mbox{and} \quad G(x)>G_{H}(x)  \quad \mbox{on} \quad
(a,\infty).
\end{equation}
If necessary, switching the order of $G$ and $G_{H}$ we can always
assume that condition \eqref{case_one} holds. Notice, that the case $G(x) \equiv G_{H}(x)$ is obviously
excluded, since $H(\xi)$ is not identically zero.
Further, due to Proposition \ref{prop:asequiv}, $G$ and $G_{H}$ share
the same asymptotic behavior at infinity:
$G(x)\sim G_{q}(x), 
\, x \to \infty. $
Since $G$ and $G_{H}$ are symmetric about the origin, it suffices to
consider these functions only for $x\ge 0.$ Furthermore, it follows from \eqref{case_one} that
$G^{2q-1}(x)>G_{H}^{2q-1}(x)$ on the interval $[0,a),$ and $G^{2q-1}(x)<G_{H}^{2q-1}(x)$ on the interval $(a,\infty).$

Consider the function
$\omega(x)=a \left[ G^{2q-1}(ax)-G_{H}^{2q-1}(ax) \right].$
This function $\omega$ is continuous by construction. Moreover,
$\omega(1)=0,$ $\omega(x)>0$ if $x\in (0,1),$ and $\omega(x)<0$ if
$x>1.$ The existence of finite $(2q-1)$-variances of $X$
and $X_H$ implies that $\int_1^{\infty}x^2|\omega(x)|dx <\infty.$
The calculations below, where the symmetry of densities are taken
into account, show that $\omega(x)$ satisfies condition $(c)$ of
Lemma \ref{lemma:omega} as well:
\begin{align*}
2\int_0^1\omega(x)dx&=
a\int_{-1}^1 \left(G^{2q-1}(ax)-G_{H}^{2q-1}(ax)\right)dx=\int_{-a}^a \left(G^{2q-1}(x)-G_{H}^{2q-1}(x)\right)dx\\
&=\int_{-a}^aG^{2q-1}(x)dx-\int_{-a}^aG_{H}^{2q-1}(x)dx\\
&=\nu_{2q-1}(X)-\int_{|x|\ge
a}G^{2q-1}(x)dx
- \left[ \nu_{2q-1}(X_H)-\int_{|x|\ge a}G_{H}^{2q-1}(x)dx \right] \\
&= \int_{|x|\ge
a} \left| G^{2q-1}(x)-G_{H}^{2q-1}(x) \right|dx=2\int_1^{\infty}|\omega(x)|dx.
\end{align*}
Here we have taken into account the equality $\nu_{2q-1}(X)=\nu_{2q-1}(X_H).$ It follows from Lemma \ref{lemma:omega} that
\begin{equation}\label{est100}
\int_0^1x^2\omega(x)dx-\int_1^{\infty}x^2|\omega(x)|dx<0,
\end{equation}
which is equivalent to
\begin{equation}\label{est10100}
\int_0^a x^2 \left( G^{2q-1}(x)-G_{H}^{2q-1}(x) \right)dx-\int_a^{\infty}x^2 \left( G_{H}^{2q-1}(x)-G^{2q-1}(x) \right)dx<0.
\end{equation}
Inequality \eqref{est10100} is the same as
$QV(X)<QV(X_H).$
Switching the order of $G$ and $G_{H}$ in the above
analysis one can see that
\eqref{case_two}  implies $QV(X)>QV(X_H).$
Both obtained relations contradict to equality $QV(X)=QV(X_H).$ Hence, our assumption on $H(\xi)$ is wrong. Thus, we conclude that $H(\xi) \equiv 0.$

The necessity is obvious, since $H(\xi) \equiv 0$ immediately implies $G_{H}= G,$ which consequently yielding $\sigma_{2q-1}^2(X)=\sigma_{2q-1}^2(X_H).$ 
\end{proof}

\begin{thm} 
\label{thm:uniqueness} 
Let $X_N$ be a $q$-independent and identically distributed random variables 
with zero $q$-mean and finite quasivariance. Then the sequence
$Z_N$ defined in \eqref{ZN} has the unique limit distribution.
\end{thm}

\begin{proof} The existence of a limit distribution was proved in
\cite{UmarovTsallisSteinberg2008}. Suppose that there are two limit
distributions $Z_{\infty}$ and $Z_{H}$ of the sequence $Z_N$ with respective distinct
densities $G(x)$ and $G_H(x).$ Due to \eqref{qv15}, both
distributions have the same quasivariance
\begin{equation} \label{qv100}
QV(Z_{\infty})=QV(Z_{H})=QV(X_1).
\end{equation}
Moreover, due to condition \eqref{nu_hyp}, 
\begin{align*}
\sigma_{2q-1}^2(Z_N) &= \sigma_{2q-1}^2 \left( \frac{X_1+\dots + X_N}{N^{\frac{1}{2(2-q)}}} \right) 
=\frac{1}{N^{\frac{1}{2-q}}} \sigma_{2q-1}^2 (X_1+\dots + X_N)
\\
&= \frac{1}{N^{\frac{1}{2-q}}} \frac{N QV(X_1)}{\nu_{2q-1}(X_1+\dots+X_N)}
\\
&= \frac{1}{N^{\frac{q-1}{2-q}}} \frac{ QV(X_1)}{\nu_{2q-1}(X_1+\dots+X_N)}  \ \to \ C QV(X_1), \quad \mbox{as} \ N \to \infty,
\end{align*}
where $C$ is a positive constant.
This yields that $\sigma_{2q-1}^2(Z_{\infty})= \sigma_{2q-1}^2(Z_H)=C \nu_{2q-1}(X_1) \sigma_{2q-1}(X_1).$ 
Hence, all the conditions of Lemma \ref{lemma_equality} are satisfied. Thus, $H(\xi) \equiv 0,$ which implies $Z_{\infty}=Z_H,$ that is the uniqueness of the limit distribution.
\end{proof}

\section{Conclusion}
Concluding, we note that with the present results, the gap detected by Hilhorst \cite{Hilhorst2010,Hilhorst2009} in the $q$-Central Limit Theorem \cite{UmarovTsallisSteinberg2008}, has been 
adequately filled. Naturally, this does not imply that other, more general, theorems can not be thought of. For example, the requirement of strict $q$-independence for all $ N$ can obviously be released, by only requiring asymptotic $q$-independence in the $N \to \infty$ limit. It might also be possible theorems similar to Lyapunov-Lindeberg type theorems \cite{Billingsley1995}, or $q$-versions of CLT for weakly dependent random variables with various mixing conditions \cite{Billingsley1995,Peligrad1986,DehlingDenkerPhilipp1986}. 
{The $q$-CLT assumes the finiteness of the quasivariance $QV(X)<\infty.$ The uniqueness of the limiting $q$-L\'evy processes studied in \cite{UmarovTsallisGellMannSteinberg2010} which corresponds to the case $QV(X)=\infty,$ is also a challenging problem. }
Moreover, at the present stage, 
we can not strictly refute existence of dependencies between the $N$ random variables other than $q$-independence, that could also exhibit $q$-Gaussians as attractors in the space of probability distributions. Further efforts along these lines are of course welcome.

\section*{Acknowledgments}
We are thankful to H.J. Hilhorst for initiating discussions on the uniqueness of the limiting distribution in the $q$-CLT, and communicating to us his counterexamples prior to publication. We also
acknowledge worthful remarks from M. Hahn, X. Jiang, and K.P. Nelson. One of us (CT) has  benefitted from partial financial support by CNPq and Faperj (Brazilian agencies), as well as by the John Templeton Foundation (USA).

\end{document}